\documentclass[letterpaper, 10 pt, conference]{ieeeconf}  

\IEEEoverridecommandlockouts                              

\usepackage{amsmath} 
\usepackage{amssymb}  

\usepackage{amsthm}  
\usepackage[short]{optidef}	
\usepackage{xfrac}
\usepackage{pgfplots}
\usepgfplotslibrary{patchplots}
\usetikzlibrary{external}
\newlength\figureheight
\newlength\figurewidth

\usepackage{easy-todo}
\let\labelindent\relax
\usepackage{enumitem}

\theoremstyle{plain}
\newtheorem{theorem}{Theorem}
\newtheorem{lemma}{Lemma}
\newtheorem{corollary}{Corollary}
\theoremstyle{definition}
\newtheorem{definition}{Definition}
\newtheorem{assumption}{Assumption}

\theoremstyle{remark}
\newtheorem{remark}{Remark}

\newcommand{\set}[1]{\mathcal{#1}}
\newcommand{\dist}{\mathcal{Q}} 
\newcommand{\setst}[2]{\left\{ #1 \, \middle| \, #2 \right\}}

\newcommand{\tp}{{\mkern-1.5mu\mathsf{T}}}

\DeclareMathOperator{\E}{\mathbb{E}}
\DeclareMathOperator{\var}{\text{var}}
\DeclareMathOperator{\diag}{\text{diag}}

\title{\LARGE \bf
Stochastic Model Predictive Control for Linear Systems using Probabilistic Reachable Sets
}

\author{Lukas Hewing, Melanie N. Zeilinger
\thanks{This work was supported by the Swiss National Science Foundation under grant no. PP00P2 157601 / 1.}
\thanks{All authors are with the Institute for Dynamic Systems and Control, ETH Z\"urich.
        {\tt\footnotesize [lhewing|mzeilinger]@ethz.ch}}%
}

\begin{document}

\maketitle
\thispagestyle{empty}
\pagestyle{empty}

\begin{abstract}
  In this paper, we propose a stochastic model predictive control (MPC)
  algorithm for linear discrete-time systems affected by
  possibly unbounded additive disturbances and subject to probabilistic constraints.
  Constraints are treated in analogy to robust MPC using a
  constraint tightening based on the concept of probabilistic reachable sets,
  which is shown to provide closed-loop fulfillment of chance constraints
  under a unimodality assumption on the disturbance distribution.
  A control scheme reverting to a backup solution from a previous
  time step in case of infeasibility is proposed,
  for which an asymptotic average performance bound is derived.
  Two examples illustrate the approach, highlighting closed-loop
  chance constraint satisfaction and the benefits of the proposed controller
  in the presence of unmodeled disturbances.
\end{abstract}
%
%
\section{Introduction}
\label{sec:introduction}
%
Robust model predictive control (MPC) methods are well-established
for dealing with bounded disturbances in a principled way~\cite{Bemporad1999a}.
For some problems, however, more detailed information about
the disturbance is available, e.g.\ in terms of a probability distribution.
Moreover, if the considered disturbance distribution has infinite support,
e.g.\ the commonly employed Gaussian distribution,
there does not exist a finite upper bound on the disturbance realizations,
limiting the applicability of robust approaches.
These observations motivate stochastic MPC methods,
which enable a potentially less conservative treatment of uncertainties
by taking knowledge of the distributions into account~\cite{Mesbah2016}.

Stochastic MPC methods can be classified into two main categories~\cite{Farina2016}:
\emph{randomized} approaches rely on the generation of a sufficient number
of disturbance realizations or \emph{scenarios},
whereas \emph{analytic approximation} methods
reformulate the problem in a deterministic form.
In this paper, we focus on the latter and propose an analytic approximation
method for linear time-invariant (LTI) systems under additive disturbances.
Previous work includes approaches based on stochastic tubes~\cite{Cannon2011},
or using a constraint tightening~\cite{Kouvaritakis2010, Korda2011},
some of which have recently been unified in~\cite{Lorenzen2017}.
These techniques rely on boundedness of the disturbances
in order to establish recursive feasibility, but enable a less conservative
tightening of constraints which only need to hold in probability.
Disturbance distributions of infinite support were in turn considered e.g.\
in~\cite{Cannon2008, Cannon2009, Farina2013, Farina2015} and~\cite{Paulson2017}.
The techniques typically rely on \emph{backup solutions}
in case the original MPC problem becomes infeasible.
In the case of~\cite{Cannon2008, Cannon2009} this is achieved by solving
an optimization problem with the objective of reducing constraint violations.
In~\cite{Farina2013, Farina2015} the MPC problem is instead initialized
at a specific state guaranteeing feasibility,
whereas~\cite{Paulson2017} considers a soft constrained formulation.

This paper presents a stochastic MPC approach for general disturbance
distributions with possibly infinite support using
probabilistic reachable sets (PRS) for constraint tightening,
as well as a control scheme for ensuring recursive feasibility,
for which a noise-dependent bound on the closed-loop cost can be derived.
The PRS serve a similar purpose as robust invariant sets in tube-based robust 
MPC and offer a flexible framework for stochastic MPC,
which allows for the consideration of general disturbance distributions and 
constraint sets.
The resulting stochastic MPC method inherently guarantees
a weak form of chance constraint satisfaction,
as e.g.\ used in previous approaches~\cite{Farina2013, Farina2015},
which we call \emph{predictive} satisfaction.
Under a unimodality assumption on the disturbance distribution and
for symmetric PRS, the method is shown to also guarantee
chance constraint satisfaction in a stronger sense,
termed \emph{closed-loop} satisfaction, which was not shown for
previous approaches~\mbox{\cite{Cannon2008}-\cite{Paulson2017}}.

Potentially unbounded disturbances can lead to
feasibility problems if the MPC is initialized at the currently measured state
$x(k)$, which we handle by choosing a suitable backup initialization.
The concept is similar to the approach in~\cite{Farina2013,Farina2015},
but applies the backup scheme only in case of infeasibility
without any further requirements, e.g.\ on a cost decrease.
We derive an asymptotic average cost bound for the resulting MPC controller,
providing a notion of convergence and stability in closed-loop,
and show in simulation examples that this update scheme offers advantages
over updates conditional on an additional cost decrease.

The paper is organized as follows.
Section~\ref{sec:preliminiaries} states the considered system to be controlled
and reviews notions of multivariate unimodality as relevant to
the presented approach. Section~\ref{sec:ReachableSets} introduces
the concept of probabilistic reachable sets, which forms the basis of
the stochastic MPC approach presented in Section~\ref{sec:MPC}.
Simulation examples are given in Section~\ref{sec:examples}
and the paper ends with concluding remarks in Section~\ref{sec:conclusions}.
%
\section{Preliminaries}
\label{sec:preliminiaries}
%
\subsection{Notation}
\label{subsec:Notation}
We refer to quantities of the system realized in closed-loop at time $k$
using parentheses, e.g.\ $x(k)$ is the state measured at time step $k$,
while quantities used in the MPC prediction are indexed with subscript,
e.g.\ $x_i$ is the system state predicted $i$ time steps ahead.
In order to specify the time at which the prediction is made, we use $x_i(k)$.
The weighted 2-norm is $\Vert x \Vert_P = \sqrt{x^\tp P x}$,
and $P \succ 0$ refers to a positive definite matrix.
The notation $ \set{A} \ominus \set{B} =
\setst{ a \in \set{A}}{a+b \in \set{A} \ \forall b \in \set{B}}$
refers to the Pontryagin set difference.
The distribution $\dist$ of a random variable $x$ is specified as $x \sim \dist$,
probabilities and conditional probabilities are denoted $\Pr(A)$, $\Pr(A \, | \, B)$
and the expected value and variance are $\E(x)$ and $\var(x)$, respectively.
%
\subsection{Considered System}
\label{subsec:ConsideredSystem}
%
We consider the problem of regulating an LTI system
subject to additive disturbances
\begin{align} \label{eq:system}
x(k\!+\!1) &= A x(k) + Bu(k) + w(k) \, ,
\end{align}
with state $x(k) \in \mathbb{R}^{n_x}$, inputs $u(k) \in \mathbb{R}^{n_u}$ and
disturbances $w(k) \in \mathbb{R}^{n_x}$,
which are assumed to be i.i.d.\ with distribution $w(k) \sim  \dist^w$.
The system is subject to chance constraints on both the states and inputs, i.e.\
\begin{subequations}\label{eq:chanceConstraints}
  \begin{align}
    \Pr(x(k) \in  \set{X}) \geq p_x \, ,\\
    \Pr(u(k) \in  \set{U}) \geq p_u \, ,
  \end{align}
\end{subequations}
where $\set{X}$ and $\set{U}$ are convex sets containing the origin.
Throughout the paper, the initial state of the system is considered known, 
such that the probabilities are conditional given the initial state, which is, 
however, omitted for simplicity of notation.
Note that~\eqref{eq:chanceConstraints} includes the case of hard constraints,
e.g.\ on the inputs, by imposing a probability of~1.
In general, however, hard constraints can only be satisfied
if the disturbance distribution has bounded support.

For the majority of results in this paper, we require no assumptions
on the nature of the disturbance distribution $\dist^w$.
In order to guarantee satisfaction of~\eqref{eq:chanceConstraints}
for the closed-loop system (Section~\ref{subsec:constraintTightening}), however,
we require the disturbance distribution to be multivariate unimodal,
the main properties of which are summarized in the following.
%
\subsection{Multivariate Unimodality}
\label{subsec:unimodality}
%
\begin{definition}[Monotone Unimodality~\cite{Dharmadhikari1976}]
 A distribution $\dist$ in $\mathbb{R}^{n_x}$ is called monotone unimodal
 if for every symmetric convex set $\set{R} \subset \mathbb{R}^{n_x}$
 and every nonzero $x \in \mathbb{R}^{n_x}$ the probability
 $\nobreak{\Pr(w+kx \in \set{R})}$ with $w\sim \dist$
 is non-increasing in $k \in [0,\infty]$.
\end{definition}
This property similarly holds if $x$ is a random variable.
\begin{lemma}\label{lm:conv}
Let the random variables $w$ and $x$ be independent
and the distribution of $w$ be monotone unimodal, then
\[
\Pr(w \in \set{R}) \geq \Pr(w + x \in \set{R}) \, ,
\]
for any convex symmetric set $\set{R}$.
\end{lemma}
\begin{proof}
  See Appendix.
\end{proof}
A related, but stronger, notion of multivariate unimodality is
\emph{central convex unimodality}.
\begin{definition}[Central Convex Unimodality~\cite{Dharmadhikari1976}]
  A distribution $\dist$ in $\mathbb{R}^{n_x}$ is called central convex unimodal
  if it is in the closed convex hull of the set of all uniform distributions
  on symmetric compact convex bodies in $\mathbb{R}^{n_x}$.
\end{definition}
\begin{theorem}[\cite{Dharmadhikari1976}]\label{thm:conv-monotone}
  Every central convex unimodal distribution is monotone unimodal.
\end{theorem}
Additionally, central convex unimodal distributions are closed under linear
transformation, convolution with another central convex unimodal distribution
and marginalization~\cite{dharmadhikari1988}.
A prominent family of distributions that are central convex unimodal
are log-concave distributions.
\begin{definition}[Log-concave Distribution~\cite{Saumard2014}]
A distribution $\dist$ in $\mathbb{R}^{n_x}$ is called log-concave,
if its probability density function is given by $f = \exp(\phi)$,
where $\phi$ is a concave function.
\end{definition}
\begin{theorem}[\cite{dharmadhikari1988}]
  Every centrally symmetric, absolutely continuous log-concave distribution
  is central convex unimodal.
\end{theorem}
Log-concave distributions are closed under affine transformation,
truncation over convex sets and marginalization~\cite{Saumard2014}.
\begin{remark}
The class of log-concave distributions is fairly rich and, e.g.,
includes multivariate Gaussian distributions.
\end{remark}
%
%
\section{Probabilistic Reachable Sets}
\label{sec:ReachableSets}
%
In order to satisfy the chance constraints~\eqref{eq:chanceConstraints},
we make use of probabilistic analogies of robust reachable sets and
MPC techniques based on constraint tightening.
For defining the required components and their properties,
consider an autonomous LTI system under additive disturbances
\begin{equation}
  x(k\!+\!1) = A_K x(k) + w(k) \, , \label{eq:sysLinError}
\end{equation}
with $x(k) \in \mathbb{R}^{n_x}$, i.i.d. $w(k) \sim \dist$ and stable matrix $A_K$,
for which we define the following probabilistic notions of reachability.
%
\subsection{Definitions}
\label{subsec:reachSetDefinition}
%
\begin{definition}[Probabilistic $n$-step Reachable Set]
A set $\set{R}^n$ with $n \geq 0$ is said to be a probabilistic
$n$-step reachable set ($n$-step PRS) of probability level $p$
for system~\eqref{eq:sysLinError} if
\[
	x(0) = 0 \Rightarrow \Pr(x(n) \in \set{R}^n) \geq p\, .
\]
\end{definition}
\begin{definition}[Probabilistic Reachable Set\label{def:PRS}]
	A set $\set{R}$ is said to be a probabilistic reachable set (PRS)
  of probability level $p$ for system~\eqref{eq:sysLinError} if
	\[
	x(0) = 0 \Rightarrow \Pr(x(n) \in \set{R}) \geq p\ \ \forall n \geq 0. \,
	\]
\end{definition}
From these definitions it follows that a PRS can be obtained from
\begin{equation}
  \set{R} = \bigcup\limits_{n=1}^{\infty} \set{R}^n \, . \label{eq:Runion}
\end{equation}
For many disturbance distributions, the $n$-step reachable set satisfies
a nestedness property, which simplifies the computation
according to $\eqref{eq:Runion}$ as outlined below.
\subsection{Nestedness}
\label{subsec:Nestedness}
It is well-known that for LTI systems the infinite-time robust reachable set
with initial state at the origin coincides with the minimal robust invariant
set~\cite{Blanchini1999} and that the sequence of reachable sets is nested,
i.e.~the $n\!-\!1$-step reachable set is a subset of the $n$-step reachable set.
In the stochastic setting, these properties in general do not hold.
Under the assumption that the disturbance follows a
central convex unimodal distribution, however,
we can recover a similar nestedness result for probabilistic reachable sets.
\begin{lemma}\label{lm:nestedness}
If $\dist$ is central convex unimodal,
any convex symmetric $n$-step PRS $\set{R}^n$ is also an $n\!-\!1$-step PRS\@.
\end{lemma}
\begin{proof}
Since central convex unimodal distributions are closed
under linear transformation and convolution, we have with $x(0) = 0$ that
\begin{equation}\label{eq:nestedness}
   x(n) = \sum_{i=0}^{n-1} A_K^{n-i-1} w(i) = A_K^{n-1}w(0)
          + \sum_{i=1}^{n-1} A_K^{n-i-1} w(i)
\end{equation}
is central convex unimodal and, by Theorem~\ref{thm:conv-monotone},
monotone unimodal for all $n\geq1$.
We similarly have
\[ x(n\!-\!1) = \sum_{i=0}^{n-2} A_K^{n-i-2} w(i) = \sum_{i=1}^{n-1} A_K^{n-i-1} w(i\!-\!1) \, .\]
Since $x(n\!-\!1)$ has the same distribution as the last term
in~\eqref{eq:nestedness}, we can use Lemma~\ref{lm:conv} and get
\[
  \pushQED{\qed}
  \Pr(x(n) \in \set{R}^n) \leq \Pr(x(n\!-\!1) \in \set{R}^n) \, . \qedhere
\]
\end{proof}
\begin{remark}
Under the assumption of central convex unimodality, $\set{R}$ can thus be
directly obtained without taking iterations via $i$-step PRS in~\eqref{eq:Runion},
i.e. $\set{R} = \lim_{n\rightarrow \infty} \set{R}^n$,
and can be approximated using Markov chain Monte Carlo methods.
\end{remark}
%
\subsection{Variance-based PRS Construction}
\label{subsec:variancePRS}
A popular way to construct a PRS is by tracking mean and variance of
$x(k)$ in~\eqref{eq:sysLinError}, which are given by
\begin{align*}
  \E(x(k\!+\!1)) &= A_K \E(x(k)) + \E(w(k)) \, ,\\
  \var(x(k\!+\!1)) &= A_K\var(x(k))A_K^\tp + \var(w(k)) \, .
\end{align*}
Applying the Chebyshev bound provides that
\begin{equation}\label{eq:chebychevPRS}
  \set{R}_c^n := \setst{x}{(x - \E{(x(n))}^\tp \var{(x(n))}^{-1} (x - \E(x(n)) \leq \tilde{p} }
\end{equation}
is an $n$-step PRS of probability level $p = 1- n_x/\tilde{p}$.

Assuming that the disturbance distribution has zero mean,
these sets similarly satisfy the nestedness property of Lemma~\ref{lm:nestedness}.
\begin{lemma}[Chebyshev Reachable Set]\label{lm:chebychevReach}
Let $\E(w(k)) = 0$. The set $\set{R}_c^n$ in~\eqref{eq:chebychevPRS}
is an $i$-step PRS of probability level $\nobreak{p= 1- n_x/\tilde{p}}$
for all $0\leq i \leq n$.

In particular,
$\nobreak{\set{R}_c := \setst{e}{e^\tp \Sigma^{-1}_\infty e \leq \tilde{p}}}$,
where $\Sigma_\infty$ solves the Lyapunov equation
$A_K \Sigma_\infty A_K^\tp - \Sigma_\infty = -\var(w(k))$
is an $i$-step PRS of level $\nobreak{p= 1- n_x/\tilde{p}}$ for any $i \geq 0$.
\end{lemma}
\begin{proof}
  The claim follows from straightforward application of the multivariate
  Chebyshev inequality and the fact that the sets are nested, i.e.\
  $\set{R}_c^n \subseteq \set{R}_c^{n+1}$, see~\cite{Hewing2018} for related results.
\end{proof}
\begin{remark}\label{rm:GaussDist}
  If $w(k)$ is normally distributed, $\set{R}_c^n$ with
  $\tilde{p} = \chi^2_{n_x}(p)$, is an $n$-step PRS of probability level $p$,
  where $\chi^2_{n_x}(p)$ is the quantile function of the chi-squared
  distribution with $n_x$ degrees of freedom.
\end{remark}
\section{Stochastic MPC using Probabilistic Reachable Sets}
\label{sec:MPC}
%
In the following, we present a stochastic MPC approach for LTI systems making
use of the concept of probabilistic reachable sets for constraint tightening.
We split the system state $x(k)$ into a nominal and error part
\[
x(k) = z(k) + e(k)
\]
with the intent to design a nominal MPC controller for $z(k)$.
Similar to robust tube-based MPC~\cite{Rawlings2009}, we keep the error $e(k)$
in a neighborhood of the nominal trajectory by using an
auxiliary state feedback controller $K$,
such that the input to system~\eqref{eq:system} is given by
\begin{equation}\label{eq:appliedInput}
  u(k) = v(k) + Ke(k) \, ,
\end{equation}
where $v(k)$ is the nominal input from the MPC for $z(k)$.
The chance constraints on uncertain states and inputs in \eqref{eq:chanceConstraints}
are then reformulated w.r.t. PRS on the error,
implementing conditions of the form $\Pr(e(k) \in \set{R}) \geq p \ \forall k$.

The proposed control scheme is characterized by the central idea
that $z(k) = x(k)$ should be selected whenever possible to introduce feedback
on $z(k)$ from measurements and react to unmodeled disturbances.
Due to the possible unboundedness of the disturbance $w(k)$, this can, however,
lead to infeasibility of the optimization problem,
in which case $z(k)$ is chosen by a backup strategy.
Similar concepts have been proposed in~\cite{Farina2013, Farina2015},
where the choice of $z(k) = x(k)$ is subject to additional conditions
related to a Lyapunov decrease in order to guarantee stability,
or~\cite{Cannon2009}, where application of the backup strategy is based on
the containment in a probabilistic invariant set based on a linear control law.
In contrast, we update the nominal system state to
$z(k) = x(k)$ whenever feasible, increasing the effect of feedback
on the nominal state, while still allowing for an asymptotic cost bound.
%
\subsection{Prediction Dynamics}
\label{subsec:Regulator}
%
The proposed stochastic MPC approach relies on predictions over
a finite time horizon using linear dynamics.
These predictions do not coincide with the closed-loop trajectory
of system~\eqref{eq:system} but have the same open-loop dynamics, i.e.
\[x_{i+1} = Ax_i + Bu_i + w_i \]
where $w_i$ is also i.i.d.\ with $w_i \sim \dist^w$.
By similarly decoupling the nominal state and error, $x_{i} = z_{i} + e_{i}$
and considering $u_i = v_i + Ke_i$, the prediction dynamics become
\begin{subequations}
\begin{align}
z_{i+1} &= A z_{i} + B v_{i} \, , \\
e_{i+1} &= (A+BK) e_{i} + w_{i} \, , \label{eq:errorSystem}
\end{align}
\end{subequations}
where the nominal predicted system state $z_i$ is deterministic,
while the predicted error $e_i$ is a random variable.

We use the predictions of the nominal system state $z_i$
to define a nominal MPC problem,
while the predicted error $e_i$ is essential for constraint tightening and
analysis of chance constraint satisfaction
(Section~\ref{subsec:constraintTightening}).
%
\subsection{Stochastic MPC Formulation \& Conditional Update}
\label{subsec:NominalMPC}
The stochastic MPC controller can be formulated using a
deterministic MPC optimization problem for the nominal system
\begin{mini!}[4]
	{Z,V}{\Vert z_N \Vert_{Q_f}^2 + \sum_{i=0}^{N-1} \Vert z_i \Vert_Q^2 + \Vert v_i \Vert_R^2}
	{\label{eq:MPC}}{}
  \addConstraint{z_{i+1}}{= A z_i + B v_i}
  \addConstraint{z_{N}}{\in \set{Z}_f}
  \addConstraint{z_{i}}{\in \set{Z}}
  \addConstraint{v_{i}}{\in \set{V}}
	\addConstraint{z_0}{= z(k)}
\end{mini!}
for all $i \in \{1,\ldots,N\!-\!1\}$ with state and input sequence
$Z = \{z_0, \ldots, z_N\}$, $V = \{v_0,\ldots,v_N\}$,
 a quadratic cost function with $Q_f,Q, R \succ 0$,
 as well as suitably tightened constraints
 $\set{Z} \subseteq \set{X}$, $\set{V} \subseteq \set{U}$,
 which will be detailed in Section~\ref{subsec:constraintTightening}.
We consider a terminal set $\set{Z}_f \subseteq \set{Z}$,
which is subject to the usual requirements, i.e.\
it is a positive invariant set under the local control law $v_i = Kz_i$,
which satisfies the input constraints
$Kz_i \in \set{V} \ \forall z_i \in \set{Z}_f$
and yields the cost decrease
\begin{equation} \label{eq:nominal_decrease}
    \Vert Az_i + BKz_i \Vert_{Q_f}^2 - \Vert z_{i} \Vert_{Q_f}^2
    \leq -\Vert z_i \Vert_Q^2 -\Vert v_i \Vert_R^2 \, \forall z_i \in \set{Z}_f \, .
\end{equation}
The nominal input applied in~\eqref{eq:appliedInput}
is $v(k) = v_0^*(z(k))$, i.e.\ the first element of the optimal input sequence
obtained from~\eqref{eq:MPC}.
\begin{assumption}[Initial Feasibility]
  We assume that optimization problem~\eqref{eq:MPC} is feasible for $z(0) = x(0)$.
\end{assumption}
Different from the system state $x(k)$, the nominal system state $z(k)$
can be selected, resulting in a corresponding error $e(k)$.
Due to disturbances that might drive $x(k)$ outside of
the feasible region, the choice of $z(k) = x(k)$ is not generally possible.
An obvious alternative is to set $z(k)$ to the first nominally predicted value
from the previous time step, which we denote $z_1\!(k\!-\!1)$.
While this enables straightforward analysis of stability,
recursive feasibility and chance constraint satisfaction,
this choice is generally not desirable,
since $z(k)$ would not be influenced by the measured states $x(k)$,
hence there would no feedback on $z(k)$~\cite{Rawlings2009}.
We therefore set $z(k) = x(k)$ whenever it is feasible in
optimization problem~\eqref{eq:MPC}, which we call Mode 1 ($M^1$).
Otherwise, we choose Mode 2 ($M^2$), the backup strategy,
which sets $z(k) = z_1\!(k\!-\!1)$ and is guaranteed to be feasible.
This results in the conditional update rule
\begin{align} \label{eq:condUpdate}
  z(k) := &\begin{cases}
            x(k) &, \text{ if feasible in~\eqref{eq:MPC} } (M^1) \\
            z_1\!(k\!-\!1) &, \text{ otherwise } (M^2) \, .
          \end{cases}
\end{align}
Note that the resulting controller is not a state-feedback controller,
since it is not a function of only $x(k)$,
but rather a feedback controller in an extended state
$\nobreak{u(k) = \kappa(x(k),z_1\!(k\!-\!1))}$.
\begin{remark}
  An alternative backup strategy,
  avoiding the solution of~\eqref{eq:MPC} in Mode 2,
  is to apply the shifted solution of~\eqref{eq:MPC} from the previous time step
  $v(k) = \nobreak{v^*_1(k\!-\!1)}$, since
  $\bar{V} = \{v^*_1\!(k\!-\!1), \ldots, v^*_{N-1}\!(k\!-\!1), Kz^*_N\!(k\!-\!1)\}$
  corresponds to a feasible suboptimal solution at time step $k$.
  The results on constraint satisfaction and the average asymptotic cost
  in the following sections remain unchanged.
  We select the receding horizon optimization of the nominal trajectory
  also in Mode~2 for notational convenience
  and the fact that it is expected to improve closed-loop performance.
\end{remark}
%
\subsection{Constraint-tightening for Chance Constraint Satisfaction}
\label{subsec:constraintTightening}
%
We make use of PRS for the predicted error system~\eqref{eq:errorSystem}
according to Definition~\ref{def:PRS} in order to tighten the constraints
such that chance constraints on $x$ and $u$ are satisfied
via the deterministic constraints on $z$ and $v$.
We allow for different tightening levels of state and input constraints
to acount for the case that different probability levels are selected,
e.g.\ input constraints are often required to be fulfilled with probability~1.

This results in two PRS $\set{R}_x$ and $\set{R}_u$ for the predicted error
system \eqref{eq:errorSystem} of probability level $p_x$ and $p_u$, respectively,
with which the constraint tightening is defined as
\begin{subequations}\label{eq:tightening}
\begin{align}
  z_{i} \in \set{Z} &:= \set{X} \ominus \set{R}_x \, ,\\
  v_{i} \in \set{V} &:= \set{U} \ominus K\set{R}_u \, .
\end{align}
\end{subequations}
\begin{remark}\label{rm:individualConstraints}
Treatment of different individual constraints,
as opposed to joint constraints, can be analogously achieved by introducing a
PRS for each constraint separately.
\end{remark}
Note that neither constraint sets nor the PRS are required to be bounded,
it is therefore possible to use probabilistic reachable sets for tightening
that are unbounded in a direction that is unconstrained,
e.g.\ for the tightening of half-space constraints~\cite{Blanchini2015}.
It is generally desirable to design the PRS for tightening such that
the Pontryagin difference in~\eqref{eq:tightening} remains as big as possible.
This can be achieved by considering tight PRS,
e.g.\ in the sense of Gaussian distributions using
Lemma~\ref{lm:chebychevReach} with Remark~\ref{rm:GaussDist}, 
and choosing the sets for tightening
such that they are aligned with the constraint sets,
e.g.\ tightening of a half-space constraint by a parallel half-space PRS
based on the corresponding marginal distribution.
\begin{remark}
A less conservative tightening is possible using time-varying confidence bounds,
i.e.\ probabilistic n-step reachable sets $\set{R}^n$, while the infinite
time reachable set $\set{R}$ is used only for the terminal set $\set{Z}_f$.
For simplicity we consider the case of constant tightening by $\set{R}$.
\end{remark}
The use of a conditional update scheme~\eqref{eq:condUpdate} complicates
analysis of chance constraint satisfaction~\eqref{eq:chanceConstraints},
since the closed-loop error $e(k)$ does not follow~\eqref{eq:errorSystem}
and evolves nonlinearly.
A tightening of the constraints under the assumption of linear error propagation
in the prediction therefore does not necessarily guarantee satisfaction
of the chance constraints~\eqref{eq:chanceConstraints} in closed-loop
when used with a conditional update scheme such as~\eqref{eq:condUpdate}.

In the following, we make use of $\set{R}$ to refer to properties
relating to both $\set{R}_x$ and $\set{R}_u$ to simplify notation.
%
\subsubsection{Chance Constraint Satisfaction in Prediction}
%
As already noted in~\cite{Farina2016}, constraint tightening based on the
predicted error guarantees chance constraint satisfaction
of the predicted states, given that the optimization
problem~\eqref{eq:MPC} is feasible at $z(k) = x(k)$, i.e.\ whenever $M^1$.
From the definition of a probabilistic reachable set $\set{R}$
we have for the predicted error
\[
\Pr(e_i \in \set{R}) \geq p \ \forall i\geq 0\, ,
\]
when $e_0 = e(k) = 0$, i.e.\ in $M^1$.
Under no further assumptions on the disturbance distribution or set $\set{R}$
we can therefore only state the probabilistic guarantees:
\begin{subequations}\label{eq:pChanceConstraints}
  \begin{align}
    \Pr(x_i \in  \set{X} \, | \, M^1) \geq p_x \ \forall i\geq0 \, ,\\
    \Pr(u_i \in  \set{U} \, | \, M^1) \geq p_u \ \forall i\geq0 \, ,
  \end{align}
\end{subequations}
which are directly obtained from
$\nobreak{\Pr(e_i \in \set{R}_x)} \geq p_x$, since $z_i \in \set{Z} = \set{X} \ominus \set{R}_x$
and $\nobreak{\Pr(e_i \in \set{R}_u)} \geq p_u$, since $v_i \in \set{V} = \set{U} \ominus K\set{R}_u$
%
\subsubsection{Closed-loop Chance Constraint Satisfaction}
\label{subsubsec:CLchanceConstraints}
Satisfaction of the chance constraints~\eqref{eq:chanceConstraints}
for the closed-loop system requires that
\[
\Pr(e(k) \in \set{R}) \geq p \ \forall k\geq0 \, ,
\]
given that $e(0) = 0$,
that is the fulfillment of the constraints for the closed-loop error $e(k)$,
which has not been addressed in previous work~\cite{Farina2013, Farina2015, Paulson2017}.

Under the assumption that $\dist^w$ is central convex unimodal
and the PRS convex symmetric, the following Theorem establishes
that $\set{R}$ is a PRS for the closed-loop error $e(k)$
which implies chance constraint satisfaction for the closed-loop system.
\begin{theorem}[PRS for Closed-Loop Error]\label{lm:constraintSatisfaction}
Let $\dist^w$ be central convex unimodal and let $\set{R}$ be a convex symmetric set.
For system~\eqref{eq:system} under the control law~\eqref{eq:appliedInput}
resulting from~\eqref{eq:MPC} with tightening~\eqref{eq:tightening},
and the conditional update rule~\eqref{eq:condUpdate} we have
\[ \Pr(e(k) \in \set{R}) \geq \Pr(e_k \in \set{R})\, , \]
for all $k\geq 0$, conditioned on $e(0) = e_0 = 0$.
\end{theorem}
\begin{proof}
Let $e_i(k)$ be the error predicted $i$ steps ahead at time $k$
using the linear dynamics~\eqref{eq:errorSystem}, with $e_0(k) = e(k)$ for all
$k$. The error $e_i(k)$ therefore depends on random variables
$w(0),\ldots,w(k\!-\!1)$ through the closed-loop dynamics, as well as
$w_0(k),\ldots,w_{i-1}(k)$ through the prediction dynamics.
We prove the claim by showing that
\[
\Pr( e_{n}(k-n) \in \set{R}) \geq \Pr( e_{n+1}(k-n-1) \in \set{R})
\]
for $n = 0,\ldots,k\!-\!1$,
from which
$\Pr(e_{0}(k) \in \set{R}) \geq \Pr( e_{k}(0) \in \set{R})$
follows immediately.
We denote with $M_k^1$ and $M_k^2$ if Mode 1 or 2 was active in time step $k$
and use $A_K = A+BK$.
With $\nobreak{\tilde{e}_{n} = \sum_{i=0}^{n-1} A_K^{n-i-1} w_{i}}$
and $\tilde{e}_0 = 0$ we have
\[
  \Pr( e_{n}(k\!-\!n) \in \set{R})  = \Pr (A_K^{n} e(k\!-\!n) + \tilde{e}_{n} \in \set{R}) \, .
\]
Note that the closed-loop error $e(k)$ is equal to $0$ whenever $M^1_k$
and equal to $A_K e(k\!-\!1) + w(k-1)$,
conditioned on the fact that it leads to infeasibility, whenever $M^2_k$.
Splitting the probability based on the active mode therefore gives
\begin{align}
  &\Pr( e_{n}(k\!-\!n) \in \set{R}) \nonumber \\
   =&\Pr (A_K^{n} e(k\!-\!n) + \tilde{e}_{n} \in \set{R} \, | \, M_{k-n}^2) \Pr(M_{k-n}^2) \nonumber \\
  &+ \Pr (A_K^{n} e(k\!-\!n) + \tilde{e}_{n} \in \set{R} \, | \, M_{k-n}^1) \Pr(M_{k-n}^1) \nonumber \\
  = &\Pr (A_K^{n+1} e(k\!-\!n\!-\!1) + A_K^n w(k\!-\!n\!-\!1) \nonumber \\
  & \qquad \qquad \qquad + \tilde{e}_{n} \in \set{R} \, | \, M_{k-n}^2) \Pr(M_{k-n}^2) \nonumber \\
  &+ \Pr (\tilde{e}_{n} \in \set{R} \, | \, M_{k-n}^1) \Pr(M_{k-n}^1) \, . \label{eq:CL_PRS_proof1}
\end{align}
Since $\tilde{e}_{n}$ is independent of the other random variables
and convex unimodal, Lemma~\ref{lm:conv} allows for bounding
\begin{align*}
&\Pr (\tilde{e}_{n} \in \set{R} \, | \, M_{k-n}^1) \\ \geq &\Pr(A_K^{n+1} e(k\!-\!n\!-\!1) + A_K^n w(k\!-\!n\!-\!1) + \tilde{e}_{n} \in \set{R} \, | \, M_{k-n}^1)
\end{align*}
which, substituted in~\eqref{eq:CL_PRS_proof1}, yields
\begin{align*}
  \eqref{eq:CL_PRS_proof1} \geq &\Pr (A_K^{n+1} e(k\!-\!n\!-\!1) + A_K^n w(k\!-\!n\!-\!1) + \tilde{e}_{n}  \in \set{R}) \\
  = &\Pr (A_K^{n+1} e(k\!-\!n\!-\!1) + \tilde{e}_{n+1}  \in \set{R}) \\
  = &\Pr (e_{n+1}(k\!-\!n\!-\!1) \in \set{R}) \, ,
 \end{align*}
since $A_K^n w(k\!-\!n\!-\!1) + \tilde{e}_{n}$ has the same distribution as $\tilde{e}_{n+1}$.
\end{proof}

\begin{corollary}
 Theorem~\ref{lm:constraintSatisfaction} implies satisfaction
 of~\eqref{eq:chanceConstraints} for the closed-loop system.
\end{corollary}
\begin{proof}
  By initial feasibility, the conditional update scheme and
  optimization problem $\eqref{eq:MPC}$, we have that
  $z(k) \in \set{Z} = \set{X} \ominus \set{R}_x$ and
  $v(k) \in \set{V} = \set{U} \ominus K\set{R}_u$ for all $k \geq 0$.
  Since by Theorem~\ref{lm:constraintSatisfaction},
  $\Pr(e(k) \in \set{R}_x) \geq p_x$ and
  $\Pr(Ke(k) \in K\set{R}_u) \geq p_u$ the claim follows immediately.
\end{proof}
%
\subsection{Asymptotic Average Cost Bound}
\label{subsec:Stability}
%
In the following, we establish an asymptotic average cost bound for
the closed-loop system under the proposed stochastic MPC scheme
and conditional update rule, providing a notion of stability and convergence.
The bound is derived by using Lipschitz-type arguments on the optimal cost
of optimization problem~\eqref{eq:MPC}. For this we make use of the
following assumption.
\begin{assumption}\label{ass:boundedness}
  The set of feasible $z(k)$ in \eqref{eq:MPC} is bounded.
\end{assumption}
This assumption is usually valid, e.g. if the terminal and input constraint
sets are bounded.
Similar arguments have been previously used e.g.\ in~\cite{Lorenzen2017}.

It is well-known that the optimal cost $J^*(z)$ of a nominal MPC problem
with quadratic cost is piecewise quadratic in the state
$z$~\cite{Bemporad2002}. Together with
Assumption~\ref{ass:boundedness} this implies that there exists a constant $L$, such that
\begin{equation} \label{eq:lipschitz}
  J^*(z) + L \Vert e \Vert_2 \geq J^*(z + e) \, .
\end{equation}
\begin{theorem}[Cost Decrease]\label{lm:costDecrease} 
  Consider system~\eqref{eq:system} under the control law~\eqref{eq:appliedInput}
  resulting from~\eqref{eq:MPC} with tightening~\eqref{eq:tightening}
  and the conditional update rule~\eqref{eq:condUpdate}.
  Let $J^*(z(k))$ be the optimal cost of~\eqref{eq:MPC},
  $\nobreak{C = L/\sqrt{\lambda_{\min}(P)}}$,
  and $P$ a solution to the Lyapunov equation
  $\nobreak{{(A+BK)}^\tp P (A+BK) - P \preceq -\epsilon I}$
  for some $\epsilon > 0$.
  We have
  \begin{align*}
    &\E(J^*(z(k\!+\!1)) - J^*(z(k)) \nonumber  \\
    \leq &-\Vert z(k) \Vert_Q^2 -\Vert v(k) \Vert_R^2 - \epsilon C \Vert e(k) \Vert_P + C \E(\Vert w(k) \Vert_P) \, ,
  \end{align*}
  conditioned on $e(0) = 0$.
\end{theorem}
\begin{proof}
 See Appendix.
\end{proof}
Using the cost decrease in Theorem~\ref{lm:costDecrease} we can derive
an average asymptotic cost bound of the presented SMPC approach.
\begin{corollary}[Average Asymptotic Cost Bound]\label{cor:AsympConvergence} \hfill \\
Let $\nobreak{w\sim \dist^w}$. Theorem~\ref{lm:costDecrease} implies
\begin{align*}
 &\lim_{t \rightarrow \infty} \frac{1}{t} \sum_{k=0}^t
 \E \left( \Vert z(k) \Vert_Q^2 + \Vert u(k) \Vert_R^2
 + \epsilon C \Vert e(k) \Vert_P \right) \\ 
 & \leq C \E \left(\Vert w \Vert_P \right).
\end{align*}
\end{corollary}
\begin{proof}
We use a typical argument in stochastic MPC~\cite{Cannon2009, Lorenzen2017}:
  \begin{align*}
  0 &\leq \lim_{t \rightarrow \infty} \frac{1}{t} \E \left(J^*(z(t)) - J^*(z(0))\right) \\
  &\leq \lim_{t \rightarrow \infty} \frac{1}{t} \E \bigg(\sum_{k=0}^t \! -\Vert z(k) \Vert_Q^2 \! - \Vert u(k) \Vert_R^2 - \epsilon C \Vert e(k) \Vert_P \\ 
  & \quad \quad \quad \quad \quad \quad\! + \! C \E(\Vert w(k) \Vert_P) \bigg).
\end{align*}
With $\lim_{t \rightarrow \infty} \frac{1}{t} \sum_{k=0}^t C \E(\Vert w(k) \Vert_P) = C
\E(\Vert w \Vert_P)$ the claim follows.
\end{proof}
\section{Numerical Examples}
\label{sec:examples}
%
We demonstrate our approach and highlight some of its features
on a simple double integrator system
\begin{align*}
x(k\!+\!1) &=
  \begin{bmatrix}
    1 & 1 \\ 0 & 1
  \end{bmatrix} x(k) +
  \begin{bmatrix}
    0.5 \\ 1
  \end{bmatrix}
  u(k) + w(k) \, ,
\end{align*}
where $w(k) \sim \mathcal{N}(0,\Sigma)$ is distributed
following a normal distribution with variance $\Sigma = \diag({[0.01, \,1]}^\tp)$.
We furthermore consider chance constraints on the absolute value
of the second state, i.e.\ the velocity,
denoted with ${[x(k)]}_2$ and 
input constraints:
\begin{subequations}
\begin{align}
	&\Pr({|[x(k)]}_2| \leq 1.2) \geq 0.6 \, ,\label{eq:ex_stateConstr}\\
	&\Pr(|u(k)| \leq 6) \geq 0.9 \, .
\end{align}
\end{subequations}

%
\subsection{MPC Setup}
%
We choose state and input stage costs with
$Q = \nobreak{\diag({[0.1, \,1]}^\tp)}$, $R = 0.1$
and design the feedback controller $K$ as an LQR controller
based on the same weights.
The prediction horizon is set to $N=30$ and for simplicity the terminal set
is chosen as $\set{Z}_f = \{{[0, \, 0]}^\tp\} $.
\subsection{Reachable Set Computation}
%
Since the distribution of $w(k)$ is Gaussian, we can compute PRS $\set{R}_x$,
$\set{R}_u$ of level $p_x$ and $p_u$ based on the marginal distribution
of ${[\tilde{e}_i]}_2$ and $Ke_i$ as proposed in Lemma~\ref{lm:chebychevReach}
with Remark~\ref{rm:GaussDist}.

The resulting sets for tightening are
\begin{subequations}\label{eq:ex_PRS}
\begin{align}
  \set{R}_x &= \setst{e}{|{[e]}_2| \leq 0.95} \, , \\
  K\set{R}_u &= \setst{Ke}{ |Ke| \leq 3.2 } \, .
\end{align}
\end{subequations}
\subsection{Results}
We compare our approach, which we call SMPC-prs, to previous results presented
in~\cite{Farina2013, Farina2015} using the same fixed controller gain $K$.
The approach is conceptually similar to the one presented in this paper
and will be referred to as SMPC-c.
The main differences as relevant to the comparison are that in SMPC-c
\begin{itemize}
  \item the selection of Mode 1 and Mode 2 is based on feasibility
  and the requirement of achieving a lower cost w.r.t.\ a Lyapunov function.
  \item the constraint tightening is specified for individual half-space violations.
  \item the constraint tightening changes over the horizon based on
  the predicted variances of the error.
\end{itemize}
Constraint satisfaction in SMPC-c is provided for
the predicted errors~\cite{Farina2013,Farina2015,Farina2016}.

Since in SMPC-c chance constraints are defined on individual half-spaces,
we consider an individual tightening of the box constraints based on $p_x/2$,
such that using the union bound we enforce~\eqref{eq:ex_stateConstr}.
%
\subsubsection{Closed-loop Constraint Satisfaction}
We first illustrate the importance of Theorem~\ref{lm:constraintSatisfaction}
by showing that closed-loop constraint satisfaction~\eqref{eq:chanceConstraints}
can differ significantly from constraint satisfaction
in prediction~\eqref{eq:pChanceConstraints}.
For this purpose, we investigate the probability of violating one individual
half-space constraint, for which SMPC-c guarantees a minimum
satisfaction probability in prediction of
\[ \Pr ({[x_i]}_2 \geq -1.2 \, | \, M_1) \geq 80\% .\]
Simulating the system 500 times from initial state $x(0) = {[6, \, 0]}^\tp$
with different disturbance realizations and counting the number of violations
of this constraint results in an empirical satisfaction rate
during the first 10 time steps of $76.62\%$,
which indicates that the individual state constraint is not satisfied with the specified probability in closed-loop.

The reason can be related to the fact that the individual tightening can be interpreted as
a tightening with individual PRS for each constraint
(along Remark~\ref{rm:individualConstraints}) in the form of half-spaces.
These sets are clearly non-symmetric, such that the assumptions of
Theorem~\ref{lm:constraintSatisfaction} do not hold.
SMPC-c furthermore tightens the constraints based on a predicted error variance,
which is reset to $0$ whenever $M_1$, and can thereby only provide constraint
satisfaction guarantees in prediction.
Evaluating the same simulation runs w.r.t.\ the joint chance constraints,
corresponding to symmetric reachable sets, empirically shows that SMPC-c
fulfills the joint constraints~\eqref{eq:ex_stateConstr} in
closed-loop with a satisfaction rate of $71.52\%$,
which is significantly larger than the specified $p_x = 60\%$.
This can, however, not be systematically established,
as SMPC-c does not provide closed-loop guarantees.

In contrast, SMPC-prs with the symmetric PRS~\eqref{eq:ex_PRS} satisfies
the assumptions of Theorem~\ref{lm:constraintSatisfaction}
and therefore guarantees satisfaction of~\eqref{eq:ex_stateConstr} a-priori.
In fact, the empirical constraint satisfaction rate is $74.9\%$,
which is slightly higher than in SMPC-c, indicating that
the strong guarantees provided by Theorem~\ref{lm:constraintSatisfaction}
may come at a cost of higher conservatism.
%
\subsubsection{Unmodeled Disturbances}
%
\begin{figure}
	\center
	\setlength\figureheight{3.5cm}
	\setlength\figurewidth{7cm}
  \input{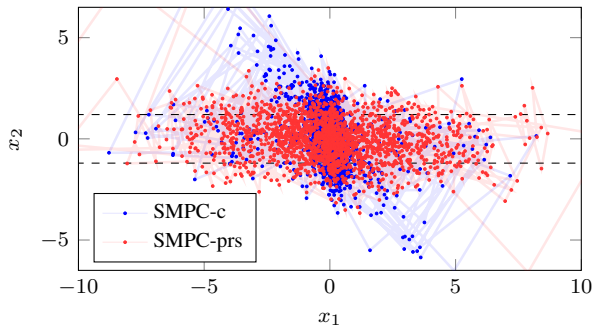}
	\caption{
  Comparison of the SMPC control approaches under unmodeled disturbances
  in every 10th time step.
  In red our approach (SMPC-prs) with update rule based on feasibility.
  In blue SMPC-c with update rule based on cost decrease.}\label{fg:ex_res}
\end{figure}
A second benefit of the proposed approach is the state feedback introduced by
the conditional update rule~\eqref{eq:condUpdate},
which can improve performance and constraint satisfaction
e.g.\ in the case of unmodeled disturbances.
To demonstrate this effect, we consider a system subject to a stronger,
unmodeled disturbance of variance $\Sigma = \diag([10, \, 1])$
at every 10th time step.
Again we compare our approach to SMPC-c, in which the nominal state
is set to the currently measured $x(k)$
only if it achieves a lower cost w.r.t.\ a Lyapunov function.

The results of the simulation are displayed in Figure~\ref{fg:ex_res}.
It is evident that SMPC-prs with its feasibility-based update rule
handles unmodeled disturbances gracefully, provided that the perturbed state
leads to a feasible optimization problem.
In the case of large disturbances, update schemes based on a Lyapunov decrease,
on the other hand, tend to apply the backup solution even if there exists a
feasible MPC solution.
As apparent in Figure~\ref{fg:ex_res} this can lead to significant
constraint violations.
In fact, in the immediate time steps after an unmodeled disturbance,
\mbox{SMPC-c} satisfies the state constraint in only $32.0\%$ of all cases,
while \mbox{SMPC-prs} does so in $72.0\%$,
satisfying the prescribed probability of $p_x = 60\%$.
%
\section{Conclusions}
\label{sec:conclusions}
%
We presented a stochastic MPC approach for LTI systems
with general additive stochastic disturbances,
which uses the concept of probabilistic reachable sets.
This enables a formulation of the MPC problem in terms of
a nominal system state with suitably tightened constraints.
Under a conditional update of the nominal system state we provided
an asymptotic average performance bound based on a cost decrease in expectation.
Results for closed-loop constraint satisfaction were presented
under the assumption that the uncertainty distribution is unimodal
and the probabilistic reachable set symmetric.
The simulation example highlights the benefits of increased feedback
provided by the proposed conditional update rule,
as well as the provided improved chance constraint satisfaction.
%

\appendices  
\section*{Appendix}
%
\begin{proof}[Proof of Lemma~\ref{lm:conv}:]
  Let $f_e$ and $f_x$ be the probability density functions of $e$ and $x$,
  respectively and $f_e * f_x$ their convolution.
\begin{align*}
  &\Pr(e + x \in \set{R}) = \int_\set{R} (f_e * f_x)(\bar{e}) d\bar{e} \\
   &= \int_{\set{R}} \int f_{e}(\bar{e}-\bar{x}) f_x(\bar{x}) d\bar{x} d\bar{e} \\
   &= \int f_x(\bar{x}) \int_{\set{R}} f_{e}(\bar{e}-\bar{x}) d\bar{e} d\bar{x} \\
   &= \int f_x(\bar{x}) \Pr(e + \bar{x} \in \set{R}) d\bar{x} \\
    &\leq \int f_x(\bar{x}) \Pr(e \in \set{R}) d\bar{x} = \Pr(e \in \set{R}) \, ,
\end{align*}
where the inequality follows from monotone unimodality.
\end{proof}
%
\begin{proof}[Proof of Theorem~\ref{lm:costDecrease}:]
Let $J(z,V)$ denote the cost of optimization problem~\eqref{eq:MPC}.
We split the expected optimal cost in cases where $M^1$ or $M^2$ apply:
\begin{align}
  &\E(J^*(z(k \! + \! 1)) \nonumber \\
    = & \E(J^*(z(k\!+\!1))|M^2) \Pr(M^2) \nonumber \\
    + &\E(J^*(z(k\!+\!1))|M^1) \Pr(M^1) \, ,\label{eq:combinedCost}
\end{align}
and find for the first term
\begin{align}
  &\E\left(J^*(z(k\!+\!1))\middle|M^2\right) = J^*(z_1\!(k)) \nonumber \\
    &\leq J(z_1\!(k), \bar{V})\, , \label{eq:costTerm1}
\end{align}
where
$\bar{V} = \{v^*_1\!(k), \ldots, v^*_{N-1}\!(k), Kz^*_N\!(k)\}$
denotes the shifted (feasible, but suboptimal) solution of the previous time step.
For the second term we have
\begin{align*}
  &\E \left(J^*(z(k\!+\!1))\middle|M^1\right) = \E \left(J^*(x(k\!+\!1))\middle|M^1\right)\\
  & \leq J^*(z_1(k)) + \E \left( L {\Vert x(k\!+\!1) - z_1(k) \Vert}_2 \middle|M^1 \right) \\
  & \leq J(z_1(k),\bar{V}) \\
  &\quad \quad \quad + \underbrace{ L/\sqrt{\lambda_{\min}(P)}}_{C} \E \left(\Vert x(k\!+\!1) - z_1(k) \Vert_P \middle|M^1 \right) \, ,
\end{align*}
where the first inequality follows from~\eqref{eq:lipschitz},
the second using the shifted suboptimal solution,
while the last uses the fact that $\lambda_{\min}(P) \Vert x \Vert^2_2 \leq \Vert x \Vert^2_P$.

Adding $C\E(\Vert x(k\!+\!1) - z_1(k)\Vert_P | M^2)$ to~\eqref{eq:costTerm1} and substituting the expressions for both modes in~\eqref{eq:combinedCost} we find
\begin{align*}
  &\E(J^*(z(k\!+\!1)) \\ \leq &J(z_1(k), \bar{V}) + C \E \left(\Vert x(k\!+\!1) - z_1(k) \Vert_P \right) \, .
\end{align*}
We can evaluate the expected value as
\begin{align*}
  \E\left(\Vert x(k\!+\!1) - z_1(k) \Vert_P \right)
  = &\E\left(\Vert (A+BK)e(k) + w(k) \Vert_P \right) \\
  \leq & \Vert (A + BK)e(k) \Vert_P + \E \left( \Vert w(k) \Vert_P \right) \\
  \leq & (1 - \epsilon) \Vert e(k) \Vert_P + \E \left( \Vert w(k) \Vert_P \right) \, ,
\end{align*}
where
$\Vert (A + BK)e(k) \Vert_P - \Vert e(k) \Vert_P \leq - \epsilon \Vert e(k) \Vert_P$
from the choice of $P$ as the solution of the Lyapunov equation.
Combining this with the usual cost decrease due to the terminal cost
and constraint in the nominal MPC~\eqref{eq:nominal_decrease}, we get
\begin{align*}
&\E(J^*(z(k\!+\!1)) - J^*(z(k)) \nonumber \\
&\leq -\Vert z(k) \Vert_Q^2 -\Vert v(k) \Vert_R^2 - \epsilon C\Vert e(k) \Vert_P + C \E(\Vert w(k) \Vert_P) \, .
\end{align*}
\end{proof}

\bibliographystyle{IEEEtran}
\bibliography{../../Bibliography/Bibliography.bib}
\end{document}